\documentclass[journal,twoside,web]{ieeecolor}
\usepackage{generic}
\usepackage{cite}
\usepackage{amsmath,amssymb,amsfonts}
\usepackage{algorithmic}
\usepackage{graphicx}
\usepackage{textcomp}

\usepackage[normalem]{ulem}
\usepackage{algorithm}

\usepackage{amsthm}

\newtheorem{proposition}{Proposition}

\usepackage{multirow}
\usepackage{array}
\usepackage{booktabs}
\let\labelindent\relax
\usepackage{enumitem}
\newcolumntype{C}[1]{>{\centering\arraybackslash}m{#1}}
\newcolumntype{?}{!{\vrule width 1.0pt}}
\bstctlcite{IEEEtran:BSTcontrol}
\newcommand{\STATEX}{\item[]}

\usepackage{etoolbox}
\makeatletter
\def\@IEEEBIOskipN{0.75\baselineskip}
\expandafter\patchcmd\csname\string\biography\endcsname
  {\vskip \@IEEEBIOskipN plus 1fil minus 0\baselineskip}
  {\vskip \@IEEEBIOskipN}
  {}
  {\PackageWarning{TCNS}{Biography spacing patch failed}}
\makeatother

\makeatletter

\renewenvironment{proof}[1][\proofname]{%
  \par
  \pushQED{\qed}%
  \normalfont
  \topsep 4\p@ \@plus 2\p@ \@minus 2\p@
  \trivlist
  \item[\hskip\labelsep\itshape #1\@addpunct{.}]%
  \ignorespaces
}{%
  \popQED
  \endtrivlist
  \@endpefalse
}
\makeatother

\def\BibTeX{{\rm B\kern-.05em{\sc i\kern-.025em b}\kern-.08em
    T\kern-.1667em\lower.7ex\hbox{E}\kern-.125emX}}
\begin{document}
\bstctlcite{BSTcontrol}
\title{Multi-Scale Datacenter Power Modulation}
\author{Akshay Sreekumar, Nicolas Christianson, Fiodar Kazhamiaka, and Ram Rajagopal
\thanks{A. Sreekumar is with the Department of Electrical Engineering, Stanford University, Stanford, CA 94305 USA. (e-mail: akshay81@stanford.edu)}
\thanks{N. Christianson was with the Department of Management Science and Engineering, Stanford University, Stanford, CA 94305 USA. He is now with the Department of Computer Science, Johns Hopkins University, Baltimore, MD 21218 USA (e-mail: christianson@jhu.edu)}
\thanks{F. Kazhamiaka is with Microsoft Azure Systems Research, Redmond WA, 98052 USA. (e-mail: fkazhamiaka@microsoft.com)}
\thanks{R. Rajagopal is with the Department of Electrical Engineering, Stanford University, Stanford, CA 94305 USA. (e-mail: ramr@stanford.edu)}}

\IEEEaftertitletext{\vspace{-4em}}
\maketitle
\thispagestyle{headings}

\begin{abstract}
Cloud datacenters must increasingly modulate power in response to time-varying grid and infrastructure constraints. 
We study this problem as finite-horizon control of a networked hybrid dynamical system, where datacenter power and service capacity depend on interactions between servers, workers, and hosted services. 
Power can be reduced through fast continuous worker throttling, which acts immediately but degrades service capacity, and slow discrete server transitions, which provide deeper savings but evolve with delay. 
Coordinating these mechanisms yields a high-dimensional mixed-integer dynamic optimization problem which is intractable to solve at scale. 
We propose a hierarchical receding-horizon controller that separates slow server reconfiguration from fast throttling recourse. For fixed server states, the throttling layer reduces to a service-level convex recourse problem solved efficiently by dual decomposition. The server layer then uses a ranked-prefix search that evaluates candidate configurations through the recourse value over the planning horizon. Experiments on realistic instances with over 15,000 servers, 200,000 workers, and 1,400 services show that the controller satisfies time-varying power caps with no violations and substantially lower service impact than fast-only or slow-only baselines. Our method offers significant speedups compared to standard optimization solvers, computing near-optimal plans within a 20 second real-time control interval. 
\end{abstract}


\section{Introduction}
\label{sec:introduction}

\IEEEPARstart{C}{loud} datacenters are among the largest single-point consumers on the power grid, with individual sites drawing tens to hundreds of megawatts.
In a 2024 survey from EPRI, $60\%$ of utilities surveyed had a datacenter interconnection request of at least $500$MW \cite{larson_utility_2024}. 
Historically, datacenters have operated under the assumption that they can consume as much power as needed up to a contracted limit.
However, this assumption is increasingly strained in the current regime of rapid load growth, where demand increases, driven in part by large-scale AI workloads, are outpacing grid infrastructure development \cite{blanford_powering_2026}. 

Operating datacenters \emph{flexibly} by dynamically modulating their power consumption offers one solution that can help to increase the utilization of existing grid resources while enabling grid operators to better manage uncertainty \cite{european_commission_flexibility_service}.
Grid instability events may demand rapid load shedding to prevent cascading blackouts \cite{skrjanc_systematic_2023}, and grid operators increasingly expect large consumers to participate in demand response and frequency regulation programs as a condition of interconnection \cite{norris_rethinking_2025}. 
Beyond grid constraints, datacenters must also contend with internal infrastructure failures such as cooling system outages, UPS faults, and power distribution equipment going down for maintenance, all of which can require immediate reduction of local power. Failing to meet the requested power limit quickly enough can trip breakers, and contribute to cascading infrastructure failures \cite{lin_slasher_2026}. While these scenarios differ in cause and significance, they all motivate the operational need for \emph{datacenters that can dynamically modulate power consumption in response to a power control signal.}

Implementing power modulation in a production cloud environment is challenging. 
A cloud datacenter is a complex, multi-tenant system where the operator has limited visibility into the workloads running on its infrastructure; ultimately, however, a datacenter's operational power flexibility is limited by performance and availability requirements in servicing these workloads.
Power is consumed by tens of thousands of servers hosting workers that collectively run a diverse population of services, reflecting the application--machine hierarchy used in production cluster managers \cite{verma_large-scale_2015}.
This can be represented as a network connecting servers, workers, and services, where actions at the server and worker level affect the capacity available to services. 
This layered, networked structure between servers, workers, and services creates a fundamental control tradeoff where power reduction actions can be taken at different levels of the network hierarchy (servers and workers), with varying impacts on response time, power savings, and service capacity.

Servers consume \emph{base power} whenever they remain powered on, even with no workers running. 
Workers cause additional power draw in proportion to their CPU frequency; we refer to this as \emph{dynamic power}.
On short timescales, the operator can modulate worker CPU frequencies, providing a fast and continuous control knob that responds immediately to power budget fluctuations. However, reducing frequency primarily affects dynamic power and quickly degrades available service capacity, making it a weak but responsive mechanism for immediate power reductions. Because server power consumption is dominated by base power, throttling workers generally cannot recover the same power as shutting servers down, even if both reduce worker capacity equally.
On longer timescales, the operator can transition servers between active and inactive states, which yields substantially larger and more persistent power savings by eliminating base power consumption. However, these actions are discrete, incur non-negligible transition delays, and must be initiated in advance of when their effects are needed.

Applying either mechanism in isolation can reduce power, but may incur excessive service impact. This motivates approaching the problem with a joint, multi-scale controller, where worker throttling handles short-term fluctuations while server reconfiguration anticipates more persistent power changes. However, this joint formulation poses two major challenges hindering deployment at datacenter scale. First, the delayed server dynamics induce a large-scale, multi-stage combinatorial control problem, even when worker throttles are fixed. Second, for any fixed server state, optimizing over worker throttles remains a large-scale convex optimization problem over the datacenter network, and this problem must be solved for multiple candidate server configurations to enable real-time, near-optimal control. Our work thus seeks to answer the following question: \emph{can we develop a principled, near-optimal controller that jointly optimizes worker throttles and server states to satisfy target power constraints while minimizing impact on service quality?} 

\subsection{Contributions}
We model datacenter power modulation as a finite-horizon, hybrid, constrained control task and address the aforementioned challenges by proposing a hierarchical receding-horizon controller. Our specific contributions are as follows:

\begin{itemize}[leftmargin=*]
    \item \textbf{Dynamic Control Formulation.} We introduce a finite-horizon control formulation of datacenter power modulation, capturing datacenter topology information, power and capacity characteristics, and service-level impact.
    \item \textbf{Hierarchical Control Architecture.} We develop a receding-horizon control scheme that decomposes the problem into slow server control and fast throttling recourse layers, enabling scalable, real-time operation.
    \item \textbf{Fast Convex Recourse Solver.} We reduce worker throttling to a lower-dimensional service-level convex recourse problem, which can be solved efficiently via dual decomposition. 
    \item \textbf{Ranked-Prefix Heuristic for Server Actuation.} We develop a ranked-prefix heuristic that scores candidate server configurations via the downstream recourse value over the planning horizon.
    \item \textbf{Empirical Results.} We provide extensive experiments evaluating the proposed controller on realistic datacenter simulations with time-varying power constraints, demonstrating the scalability and near-optimality of our algorithm.
\end{itemize}

\subsection{Related Work}
Our work contributes to a long line of prior work in datacenter operation and optimal control of hybrid systems. 

\emph{\textbf{Datacenter Power Capping.}} Power capping systems have been widely studied as a mechanism for safely operating datacenters near provisioned limits. 
Production systems such as Dynamo monitor the datacenter power hierarchy and issue server-level power caps to prevent breaker trips, while CapMaestro develops coordinated priority-aware controllers for enforcing limits across multiple levels of power infrastructure \cite{wu_dynamo_2016, li_scalable_2019}. Related systems use priority-aware capping to support power oversubscription by preferentially reducing lower-priority work, and QoS-aware capping systems such as Thunderbolt and PADS use workload-aware throttling or resource controls to enforce power budgets while limiting performance degradation \cite{sakalkar_data_2020, li_thunderbolt_2020, savasci_pads_2024}.
Slasher~\cite{lin_slasher_2026} studies an impact aware server shutdown policy, but our controller considers a broader action space that coordinates frequency scaling with server shutdowns. 

\emph{\textbf{Dynamic Capacity Management.}} A related line of work studies dynamic capacity management, where active compute capacity is adapted over time to reduce energy while maintaining application performance. Early work on dynamic provisioning for multi-tier internet applications used queuing models together with predictive and reactive mechanisms to allocate resources across application tiers under time-varying load~\cite{urgaonkar_dynamic_2005}. Server right-sizing has also been studied as an online control problem that trades off active-server operating costs against switching costs incurred when changing the number of active servers \cite{lin_dynamic_2013}, with subsequent systems such as AutoScale developing robust capacity management strategies for multi-tier datacenters \cite{gandhi_autoscale_2012}. More recent theoretical work extends right-sizing formulations to heterogeneous datacenters with discrete server decisions \cite{albers_algorithms_2021}. These works establish server provisioning as an important mechanism for datacenter energy management. In this work, control is driven by an exogenous power budget rather than workload demand alone, requiring coordination of server provisioning with throttling recourse to modulate aggregate power while minimizing service impact. 

\emph{\textbf{Datacenters as Flexible Grid Loads.}} Datacenter flexibility has also been studied as a grid-facing resource, motivated by the ability to reshape compute across time, space, and infrastructure. Prior work has examined ancillary service participation through load reduction, frequency regulation using dummy loads, and QoS-aware server power reshaping with DVFS and complementary workloads \cite{ghamkhari_data_2012, wang_frequency_2019, jahanshahi_powermorph_2022}. Chen et al. demonstrated server level tracking on an ISO regulation signal while maintaining QoS with dynamic server power capping \cite{chen_dynamic_2013}. Adjacent work has examined carbon-aware computing systems which exploit temporal and geographic flexibility to reduce emissions while meeting service requirements \cite{radovanovic_carbon-aware_2021, hall_carbon-aware_2025,lindbergGuideReducingCarbon2021,lechowiczLearningAugmentedCompetitiveAlgorithms2025b}. This large body of work underscores the grid value of flexible compute. Our work complements this literature by studying how fast worker throttling and slow server configuration can be coordinated to deliver flexibility and satisfy a time-varying power budget while limiting service impact. 

\emph{\textbf{Hybrid Control and Optimization.}} Hybrid model predictive control is a natural framework for systems with coupled continuous and discrete decisions, but the resulting online optimization problems are typically mixed-integer programs that are difficult to solve at scale \cite{borrelli_predictive_2017}. Recent work has reduced online computation through warm-starting, exploiting problem structure, and approximate solution methods \cite{marcucci_warm_2021, hespanhol_structure_2019, takapoui_simple_2020}. In parallel, decomposition methods for convex network resource-allocation problems can efficiently solve large-scale convex optimization problems \cite{chiang_layering_2007, palomar_tutorial_2006}, and have also been applied to distributed power allocation in computing clusters \cite{badiei_diba_2016}. However, existing methods do not directly address the combination of an enormous, combinatorial set of feasible server decisions, delay dynamics, and large-scale convex (nonlinear) structure in datacenter-scale power modulation. 

\section{System Model}
We consider a cloud datacenter operating under time-varying exogenous power constraints. 
At each time step, the operator must coordinate server state transitions and worker throttling decisions to ensure that the datacenter's aggregate power consumption remains within the available budget while minimizing the impact on hosted services. 
Here, we define each component of the datacenter system model to allow us to formally state the dynamic control problem in Section \ref{sec:dynamic_control_problem}.

\subsection{Datacenter Topology and Control Dynamics}
Let $\mathcal{S}$ denote the set of servers in a datacenter, and let $W_s$ denote the set of workers hosted on server $s \in \mathcal{S}$. 
The set of all workers is given by $W := \bigcup_{s \in S} W_s$. 
Each worker $w$ belongs to exactly one service $q(w) \in Q$, where $Q$ is the set of all services. 
The set of workers belonging to service $q$ is denoted $W_q := \{w \mid q(w) = q\}$. We write $s(w)$ for the server hosting worker $w$.
All workers of a given service are similar by design, since the typical "horizontal scaling" deployment pattern involves balancing load across a set of worker replicas. 
Hence, they share the same power, capacity, and load characteristics. 
Workers for a single service may reside on different servers.
We assume that worker-to-server and worker-to-service assignments are fixed over the control horizon. Thus, workers are not reassigned when their host server transitions out of the \texttt{ON} state. 

For every worker $w$ in the system, we define its corresponding frequency throttle $\alpha_w \in [\tau_w, 1]$, where $\tau_w := \tau_{q(w)} \in [0,1)$ is a service-level threshold at which the worker consumes no power and provides no useful capacity. The throttle $\alpha_w$ represents the normalized CPU frequency of the worker. Each worker's power consumption and provided capacity is a function of its throttle. At time $t$, we write $\alpha_{w,t}$ for the throttle of worker $w$, and collect all worker throttles as $\alpha_t:= (\alpha_{w,t})_{w \in \mathcal{W}}$.

The operator can also take server-level control actions, which reduce more power but evolve with delay. Specifically, the delay dynamics model a graceful shutdown procedure, where instead of abruptly losing power, the server safely persists its state before shutdown. 
For each server $s \in \mathcal{S}$ and time step $t$, let 
$x_{s,t} \in \{\texttt{ON},\texttt{OFF},\texttt{BOOT}(k),\texttt{SHUT}(k): k=1,\ldots,K\}$ 
denote the pre-action server state. 
State $\texttt{ON}$ indicates that the server is active and its hosted workers provide useful capacity, while $\texttt{OFF}$ indicates that the server is fully inactive. The transition states $\texttt{BOOT}(k)$ and $\texttt{SHUT}(k)$ represent servers in the middle of a boot or shutdown, respectively, with $k$ time slots remaining until completion. 

During these transition states, the server draws power but provides no capacity. This power draw corresponds to the idle power draw of the server itself as well as a fixed power draw from the transitioning workers. At each time $t$, the operator chooses a server action $u_{s,t}\in\{\texttt{noop},\texttt{BOOT},\texttt{SHUTDOWN}\}$ given the current server state. 
Specifically, \texttt{BOOT} actions can only be issued from the \texttt{OFF} state, \texttt{SHUTDOWN} actions can only be issued from the \texttt{ON} state, and no further actions are allowed until a server completes a transition. \texttt{noop} denotes no action is taken. 

Applying the server action yields a \emph{post-action} state $x_{s,t}^+$ that determines the realized power consumption and capacity contribution during time step $t$. We compactly write $x_t := (x_{s,t})_{s \in \mathcal S}$ and $u_t := (u_{s,t})_{s \in \mathcal S}$, and denote the collection of post-action states by $x_t^+ = \phi(x_t, u_t)$, where $\phi$ acts componentwise across servers and represents their transition dynamics. At the end of the time step, the state advances according to the countdown map $x_{t+1} = \psi(x_t^+)$, where
\begin{equation}
\psi(\texttt{BOOT}(k)) =
\begin{cases}
\texttt{BOOT}(k-1), & k=2,\dots,K,\\
\texttt{ON}, & k=1,
\end{cases}
\end{equation}


\begin{equation}    
\psi(\texttt{SHUT}(k)) =
\begin{cases}
\texttt{SHUT}(k-1), & k=2,\dots,K,\\
\texttt{OFF}, & k=1.
\end{cases}
\end{equation}
Fig.~\ref{fig:server_fsm_dynamics} illustrates the transition dynamics $\phi(\cdot)$ for a server.

\begin{figure}[t]
  \centering
  \includegraphics[width=\columnwidth]{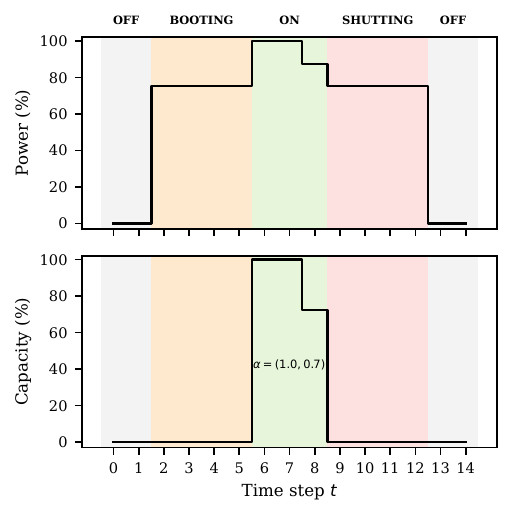}
  \caption{Representative server transition dynamics. Booting immediately increases power consumption without increasing capacity; While \texttt{ON}, power usage includes base power and throttle-dependent worker power. Two representative throttles levels ($\alpha=1 $ and $\alpha=0.7$) are shown during \texttt{ON}. Shutting down immediately decreases capacity while base power consumption decreases after a delay.}
  \label{fig:server_fsm_dynamics}
\end{figure}
The post-action state determines indicators $z^P_{s,t}=\mathbf{1}\{x^+_{s,t}\neq \texttt{OFF}\}$ and $z^C_{s,t}=\mathbf{1}\{x^+_{s,t}=\texttt{ON}\}$, denoting whether server $s$ contributes baseline power and useful capacity, respectively.
For example, $\texttt{ON}$ corresponds to $(z_{s,t}^P,z_{s,t}^C)=(1,1)$,  $\texttt{OFF}$  corresponds to $(0,0)$, and transition states correspond to $(1,0)$.

\subsection{Worker Power Model}
\label{sec:worker_power_model}
Motivated by quadratic processor-power models \cite{hager_exploring_2016}, we model each worker $w$'s power consumption as a convex quadratic function of its throttle level $\alpha_w$:

\begin{equation}
    p_w(\alpha_w) = a_w \alpha_w^2 + b_w \alpha_w, \quad \alpha_w \geq 0,
    \label{eq:worker_power}
\end{equation}
where the coefficients $a_w, b_w$ are determined by the service $q(w)$. We define the above-threshold power as 
\begin{equation}
    p_w^{\textrm{above}}(\alpha_w) := p_w(\alpha_w) - p_w(\tau_w).
    \label{eq:worker_power_above}
\end{equation}
By construction, $p_w^{\textrm{above}}(\tau_w) = 0$. 

\subsection{Worker Capacity Model}
\label{sec:worker_capacity_model}

Each worker $w$ has a concave capacity function
\begin{equation}
    g_w(\alpha_w) = \rho_w \cdot \alpha_w^{\gamma_w}, \quad \gamma_w \in (0,1),
    \label{eq:worker_capacity}
\end{equation}
where $\rho_w$ is the product of the number of cores assigned to that worker and the per-core capacity for a worker $w$ on service $q(w)$, and $\gamma_w$ captures diminishing returns from frequency scaling. The empirical characterization of application performance curves as a function of power or processor frequency is well established, in particular for datacenter applications~\cite{kanev_tradeoffs_2014, qiu_revisiting_2025, patel_characterizing_2024}. In practice, the parameters of such models can be estimated from offline service profiling, telemetry gathered from throttling events, or coarser workload-class models when service-specific measurements are unavailable. As with power, we define the above-threshold capacity as $g_w^{\textrm{above}}(\alpha_w) := g_w(\alpha_w) - g_w(\tau_w)$, so that $g_w^{\textrm{above}}(\tau_w) = 0$.

\subsection{Datacenter Power Aggregation}
\label{sec:power_aggregation}

We now describe how power consumption is aggregated from the worker level to the server and datacenter levels. 
Let $b_s$ denote the baseline power drawn by server $s$ whenever it is not $\texttt{OFF}$ (i.e., when it is \texttt{ON} or in a boot or shutdown state). 
This term captures the server-side IT power that is independent of worker-level processor throttles, such as memory, storage, and networking devices. 
Let $P_{\mathrm{tr},s}$ denote the additional power drawn while server $s$ is in a boot or shutdown transition state. 
This is a constant power defined per server that is some fixed fraction of that server's dynamic power. It represents the power used by processes that run during boot and shutdown states---e.g., during shutdown, processes that flush buffers to persistent storage and gracefully terminate running processes, or during boot, processes that perform hardware health checks, boot the OS, and prepare the server for hosting VMs. Altogether, server $s$'s power consumption at step $t$ is
\begin{equation}
\begin{split}
P_{s,t}(\alpha_t, z_{s,t}^P, z_{s,t}^C)
&=
z_{s,t}^P b_s
+
z_{s,t}^C \sum_{w \in W_s} p_w^{\textrm{above}}(\alpha_{w,t}) \\
&\quad +
\left(z_{s,t}^P - z_{s,t}^C\right) P_{\mathrm{tr},s}.
\end{split}
\label{eq:server_power_aggregation}
\end{equation}

To account for non-IT overhead such as cooling and power delivery losses, we scale the aggregate IT power by the datacenter power usage effectiveness (PUE). 
We approximate PUE as a constant factor that does not capture the delayed response of cooling power to changes in IT load. 
Let $B$ denote a fixed background power term that is independent of the control actions, including disaggregated storage servers and networking equipment.
The total datacenter power consumption at step $t$ is
\begin{equation}
P_{t}^{\mathrm{tot}}(\alpha_t, x_t^+)
=
\mathrm{PUE}
\left(
B + \sum_{s \in \mathcal{S}} P_{s,t}(\alpha_t, z_{s,t}^P, z_{s,t}^C)
\right).
\label{eq:dc_total_power}
\end{equation}
Note that the dependence on $x_t^+$ enters through $z_{s,t}^P$ and $z_{s,t}^C$.

\subsection{Service Capacity Aggregation}
\label{sec:capacity_aggregation}

Service capacity is aggregated across workers. Only workers on active servers (i.e. $\texttt{ON}$ servers) can contribute capacity.
Since workers on the same service may be distributed across multiple servers, we define capacity at the service level.

Using the capacity model from Section~\ref{sec:worker_capacity_model}, the total capacity allocated to service $q \in Q$ at step $t$ is
\begin{equation}
C_{q,t}(\alpha_t, x_t^+)
=
\sum_{w \in W_q} z_{s(w),t}^C \, g_w^{\textrm{above}}(\alpha_{w,t}).
\label{eq:service_capacity_aggregation}
\end{equation}

\subsection{Impact Model}
\label{sec:impact_model}

Following the work in \cite{lin_slasher_2026}, we define a penalty $\mathcal{I}_q(C_{q,t})$ for allocating insufficient capacity to a service $q$ as 

\begin{equation}
    \mathcal{I}_q(C_{q,t}) = \frac{\mathbb{E}[(l_q - C_{q,t})_+]}{\mathbb{E}[l_q]},
    \label{eq:impact}
\end{equation}
where $l_q$ is a random variable associated with the \emph{load} of a service $q$, expressed in the same normalized capacity units as $C_{q,t}$. We estimate the load distribution from worker CPU utilization traces. For each service, worker-level utilization time series are aggregated and binned to form an empirical load distribution, from which we draw samples. 
We approximate (\ref{eq:impact}) via the sample average by taking $N$ samples of the load:
\begin{equation}
    \mathcal{I}_q(C_{q,t}) \approx \frac{1}{N \bar{l}_q} \sum_{j=1}^N \left( l_q^{(j)} - C_{q,t} \right)_+,
    \label{eq:impact_saa}
\end{equation}
where $\bar{l}_q = \frac{1}{N}\sum_{j=1}^N l_q^{(j)}$. This function is convex and nonincreasing in $C_{q,t}$, equals zero when $C_{q,t}$ exceeds all samples, and approaches 1 as $C_{q,t} \to 0$. It provides a comparable measure of impact across services of different sizes without requiring visibility into application level performance metrics.

\subsection{Deviation Model}

For simplicity, we assume the nominal operating point of the datacenter to be the fully provisioned state in which all servers are \texttt{ON} with their workers at $\alpha=1$, but the following construction can be defined relative to any specified nominal server configuration and throttle profile. We introduce a deviation penalty that softly biases the controller back toward this nominal point once the power excursion is over, and breaks ties between equally low-impact actions by favoring configurations with fewer idle servers and less throttling. 

At each time $t$, we introduce the deviation penalty 
$D_t(\alpha_t,x_t^+) = D_t^{\mathrm{server}} + D_t^\alpha$, where 
$D_t^{\mathrm{server}}=\sum_{s\in\mathcal{S}}\big[(1-z^P_{s,t})+(1-z^C_{s,t})\big]$ 
and 
$D_t^\alpha=\sum_{\{w:z^C_{s(w),t}=1\}}(1-\alpha_{w,t})$.

\section{Dynamic Power Modulation Control Problem}

\subsection{Finite-Horizon Dynamic Control Problem}
\label{sec:dynamic_control_problem}

We now state the full dynamic power modulation control problem. At each time step $t$, the controller must choose server actions $u_t$ and worker throttles $\alpha_t$ so as to keep the datacenter within its power budget while minimizing the resulting impact on hosted services. Given selected actions $u_t$, the server states evolve via the transition dynamics $x_t^+ = \phi(x_t,u_t)$ and $x_{t+1} = \psi(x_t^+)$,
while the realized post-action state $x_t^+$ determines the power and capacity available during time step $t$ through the models developed in Sections~\ref{sec:power_aggregation} and~\ref{sec:capacity_aggregation}.

Using the service impact model from Section~\ref{sec:impact_model}, we define the stage cost at time $t$ as
\begin{equation}
J_t(\alpha_t,x_t^+)
=
\sum_{q \in Q} \theta_q \, \mathcal{I}_q\!\left(C_{q,t}(\alpha_t,x_t^+)\right) + \epsilon D_t(\alpha_t, x_t^+).
\label{eq:stage_cost}
\end{equation}
where, following \cite{lin_slasher_2026}, we define a priority weight $\theta_q \geq 0$ reflecting the operational importance of service $q$, and $\epsilon \ll 1$ is a weight for the deviation penalty. 

Let $\mathcal{H}:=\{0,\ldots,H-1\}$. At each control time $t$, given the current state $x_t$, an $H$ step power budget forecast  $\{\hat P_{t+r}^{\max}\}_{r\in \mathcal{H}}$, and load samples, we can now state the finite-horizon dynamic control problem:
\begin{subequations}
\label{prob:dynamic_control}
\begin{align}
\quad \underset{\substack{\alpha_{t:t+H-1},\\u_{t:t+H-1}}}{\text{minimize}} \quad
& \sum_{r=0}^{H-1} J_{t+r}(\alpha_{t+r},x_{t+r}^+) \label{prob:dynamic_obj} \\
\text{s.t.} \quad
& u_{s,t+r} \in \mathcal{U}(x_{s,t+r}) \quad \forall r \in \mathcal{H}, \forall s \in \mathcal{S}, \label{con:dynamic_admissible} \\
& x_{t+r}^+ = \phi(x_{t+r},u_{t+r}) \quad \forall r \in \mathcal{H}, \label{con:dynamic_post_state} \\
& x_{t+r+1} = \psi(x_{t+r}^+) \quad \forall r \in \mathcal{H}, \label{con:dynamic_next_state} \\
& P_{t+r}^{\mathrm{tot}}(\alpha_{t+r},x_{t+r}^+) \leq \hat{P}^{\max}_{t+r} \quad \forall r \in \mathcal{H}, \label{con:dynamic_power} \\
& \tau_w \leq \alpha_{w,t+r} \leq 1 \quad \forall r \in \mathcal{H}, \forall w \in \mathcal{W}, \label{con:dynamic_throttle} 
\end{align}
\end{subequations}
Problem~\eqref{prob:dynamic_control} is the $H$-step optimization solved in our model predictive control (MPC) scheme. At each time step, the controller applies only the first server action and throttle vector, updates the state and forecast, and then resolves the problem. 
Worker placement is fixed, so the discrete decisions govern only the server state and, as a result, the availability of their hosted workers. 
Constraint~\eqref{con:dynamic_admissible} limits the server actions to those admissible given the current server states. Constraints~\eqref{con:dynamic_post_state} and~\eqref{con:dynamic_next_state} enforce the server transition dynamics via $\phi(\cdot)$ and $\psi(\cdot)$. Constraint~\eqref{con:dynamic_power} enforces the datacenter power budget, while constraint~\eqref{con:dynamic_throttle} enforces worker throttle bounds. The discrete server actions are coupled across time through the delayed transition dynamics, while the continuous worker throttles affect the instantaneous power--capacity tradeoff within each slot.
\vspace{-1em}

\subsection{Stagewise Reduction and $\alpha$-Recourse Value Function}
\label{sec:recourse_reduction}

Problem~\eqref{prob:dynamic_control} optimizes jointly over server actions and worker throttles, but admits a simpler structure. Server actions determine the post-action states $x_t^+$ and are coupled across time by the transition dynamics. Worker throttles, however, have no intertemporal dynamics. For a fixed $x_t^+$, the vector $\alpha_t$ enters only through the time $t$ stage cost and power constraint. As such, given a post-action server state $x_t^+$, we can solve for the resulting optimal throttle vector $\alpha_t$ via the following convex program, which we call the $\alpha$-recourse value function:
\begin{equation}
\begin{aligned}
V_t(x_t^+)
=
\min_{\alpha_t}
\quad &
J_t(\alpha_t,x_t^+) \\
\text{s.t.}\quad
& P_t^{\mathrm{tot}}(\alpha_t,x_t^+) \leq \hat P_t^{\max}, \\
& \tau_w \leq \alpha_{w,t} \leq 1,
&& \forall w \in W.
\end{aligned}
\label{eq:stage_value_function}
\end{equation}
The quantity $V_t(x_t^+)$ is the minimum cost attainable at time $t$ by throttling once the post-action server state has been fixed.
For fixed $x_t^+$, Problem~\ref{eq:stage_value_function} is convex in $\alpha$: $P_t^{\mathrm{tot}}(\alpha_t,x_t^+)$ is convex in $\alpha_t$, while each service capacity $C_{q,t}(\alpha_t,x_t^+)$ is concave in $\alpha_t$. Since $\mathcal{I}_q(\cdot)$ is convex and nonincreasing, the composition $\mathcal{I}_q(C_{q,t}(\alpha_t,x_t^+))$ is convex in $\alpha_t$, and therefore so is the objective in \eqref{eq:stage_value_function}. The recourse problem need not be feasible for every post-action server state. We use the extended-value convention that $V_t(x_t^+)=+\infty$ whenever no feasible throttle exists. 

Using \eqref{eq:stage_value_function}, we can write an explicitly hierarchical formulation of the dynamic control problem:
\begin{equation}
\begin{aligned}
\underset{u_{t:t+H-1}}{\textrm{minimize}} \quad &
    \sum_{r=0}^{H-1} V_{t+r}(x_{t+r}^+) \\
\text{s.t.}\quad
& x_t \text{ given}, \\
& u_{s,t+r} \in \mathcal{U}(x_{s,t+r}),
&& \forall r \in \mathcal{H},\ s \in \mathcal{S}, \\
& x_{t+r}^+ = \phi(x_{t+r},u_{t+r}),
&& \forall r \in \mathcal{H}, \\
& x_{t+r+1} = \psi(x_{t+r}^+),
&& \forall r \in \mathcal{H}.
\end{aligned}
\label{eq:reduced_dynamic_control_problem}
\end{equation}
Problem~(\ref{eq:reduced_dynamic_control_problem}) suggests a natural structure to the per-time step mechanics. Once the server actions determine the post-action state $x_t^+$, the worker throttles act as fast within-step recourse, while the server decisions remain coupled across time through the delayed finite-state dynamics. However, even with this explicit two-stage separation between server actions and worker throttles, the reduced problem \eqref{eq:reduced_dynamic_control_problem} is an intractably large multi-stage discrete control problem for realistic-scale datacenters. For instance, if we have $m$ servers to control and a horizon $H$, the action space has size $\sim2^{mH}$, where in a typical instance, $mH \approx 90{,}000$.

In addition, efficiently solving the outer server-control problem (even heuristically) requires fast access to the recourse value $V_t$. At realistic datacenter scale, $V_t(x_t^+)$ is the optimal value of a convex value problem with hundreds of thousands of worker throttle variables, which can take several seconds to solve with standard solvers. Candidate server action rollouts must be scored through repeated evaluations of $V_t$ for different values of $x_t$, making a generic convex solver too slow for online control. 

To address these challenges, Section~\ref{sec:alpha_recourse} derives a fast solver for $V_t(x_t^+)$ by reducing the worker level throttling problem to a service-level convex problem, which can be solved via dual decomposition. Section~\ref{sec:ranked_prefix_controller} then uses this solver within a receding-horizon server controller, restricting the action space to a low-dimensional set of ranked-prefix configurations that can be evaluated over the planning horizon. 

\section{An Efficient Algorithm for Large-Scale $\alpha$-Recourse}
\label{sec:alpha_recourse}
In this section, we exploit the symmetry and separability of (\ref{eq:stage_value_function}) to develop an efficient and highly parallelizable method for evaluating the $\alpha$-recourse value. The resulting solver is used to score candidate server configurations in the receding-horizon controller developed in Section~\ref{sec:ranked_prefix_controller}.

\subsection{Variable Reduction}
\label{sec:variable_reduction}
We first observe that one can exploit symmetry in the objective to reduce the number of decision variables in (\ref{eq:stage_value_function}) from $\mathcal{O}(|W|)$ to $\mathcal{O}(|Q|)$---i.e., we only need to maintain a single throttle variable $\alpha_q$ \emph{per service} instead of per worker. 

\begin{proposition}
    Fix a post-action server state $x_\tau^+$, and suppose that all active workers serving the same service are identical in their throttle threshold, power model, and capacity model, and that service impact depends only on aggregate service capacity. Then, restricting all active workers of each service $q$ to operate at a single common throttle $\alpha_q$ does not change the extended optimal value $V_t(x_t^+)$. Consequently, whenever $V_t(x_t^+) < \infty$, the recourse problem admits an optimal solution in which only one throttle variable is maintained per service instead of per worker. 
\end{proposition}

\begin{proof}
   Let us consider a single service $q$, with a set of $n$ identical workers $W_q$. Assume there exists some optimal set of throttles for the workers on this service given by $(\alpha_1^\star, \alpha_2^\star, \dots, \alpha_n^\star)$. We can define $\bar{\alpha} = \frac{1}{n}\sum_{i=1}^{n}\alpha_i^\star$
   as the average throttle for the workers in service $q$. Now, we consider the symmetric solution where each of the $n$ workers is set to $\bar{\alpha}$. Observe that the symmetric solution still satisfies each worker's power throttle constraint $\alpha \in [\tau_q, 1)$ because it is a convex combination of feasible throttles.

   By Jensen's inequality, we can see that 
   \begin{equation*}
       n p^\text{above}_q(\bar{\alpha}) \leq \sum_{i=1}^n p_q^\text{above}(\alpha_i^\star)
   \end{equation*}
   and 
   \begin{equation*}
       n g^\text{above}_q(\bar{\alpha}) \geq \sum_{i=1}^n g_q^\text{above}(\alpha_i^\star),
   \end{equation*}
   meaning that the symmetric solution cannot increase worker power consumption or decrease worker capacity. Because the service impact function is assumed to be nonincreasing in the aggregate capacity, impact cannot increase from the symmetric solution. Similarly, by linearity of the deviation term, the symmetric solution does not change the deviation penalty. 

   We conclude that assigning throttle variables within each service symmetrically cannot harm feasibility or optimality. It follows that the $\alpha$-recourse problem has an optimal solution that is constant across active workers of the same service. Consequently, we may replace the worker level throttles $\alpha_w$ by a service level throttle $\alpha_q$ for each service to yield an equivalent reduced problem.  
\end{proof}
\subsection{Dual Decomposition via Bisection}

After reduction, the recourse problem has one throttle variable per service, with services coupled only through the aggregate power budget. 
The resulting problem is a separable convex resource allocation problem with a single coupling constraint, for which Lagrangian and parallel solution methods are well-established \cite{patriksson_survey_2008}.
We specialize this structure by using an outer bisection on the power price multiplier and parallel inner bisections for the service level subproblems. 
We fix an arbitrary post-action server state $x_t^+$ and drop the time index $t$ for notational convenience. Given that we now track throttles at the service level, we can write the capacity for a service $q$ with $n_q$ active workers as $C_q(\alpha_q) = n_q\rho_q(\alpha_q^{\gamma_q} - \tau_q^{\gamma_q})$, and the objective as $f_q(\alpha_q) = \theta_qI_q(C_q(\alpha_q)) +\epsilon n_q(1-\alpha_q)$.
The dynamic power associated with the service is given by $h_q(\alpha_q) = n_q\text{PUE}\big[a_q(\alpha_q^2-\tau_q^2) + b_q(\alpha_q-\tau_q)\big]$.

Because the server states are fixed, we denote all the \emph{fixed} power (from server baseline and transition power) as $P_{\text{fix}}$. We define $\tilde{R} = \hat P^{\text{max}} - P_{\text{fix}}$. The power budget constraint is then $\sum_{q \in Q} h_q(\alpha_q) \leq \tilde{R}$,
where $Q$ is the set of all services. The recourse problem can thus be written as
\begin{equation}
\begin{aligned}
\underset{\alpha_q \in [\tau_q, 1]}{\min} \quad &
\sum_{q \in Q}^{} f_q(\alpha_q) \\
\text{s.t.}\quad
& \sum_{q \in Q} h_q(\alpha_q) \leq \tilde{R}.\\
\end{aligned}
\label{eq:reduced_recourse_problem_primal}
\end{equation}
Problem \eqref{eq:reduced_recourse_problem_primal} is a convex optimization problem with a separable objective and a single coupling constraint, making it amenable to solution via dual decomposition \cite{boyd_distributed_2010}. For a fixed dual variable $\lambda \geq 0$, we define the per-service subproblem
\begin{equation}
    \eta_q(\lambda) = \underset{\alpha_q \in [\tau_q, 1]}{\min} \; f_q(\alpha_q) + \lambda h_q(\alpha_q).
\end{equation}
The dual function is then $y(\lambda) = \sum_{q \in Q} \eta_q(\lambda) - \lambda \tilde{R}$.
Evaluating $y(\lambda)$ reduces to solving $|Q|$ independent scalar minimization problems.

In standard dual decomposition, one would maximize the concave dual function $y(\lambda)$ using projected supergradient ascent. Let $\alpha_q^\star(\lambda)$ denote a primal minimizer of the per-service subproblem for a fixed $\lambda$. A supergradient of $y$ is $s(\lambda) = \sum_{q \in Q} h_q(\alpha_q^\star(\lambda)) - \tilde{R}$.
In our setting, this outer update simplifies substantially. We have only a single dual variable with the induced dynamic power nonincreasing in $\lambda$, and thus the supergradient $s(\lambda)$ is scalar and monotone. Hence, we efficiently solve the dual problem via bisection on $\lambda$.

Moreover, each per-service subproblem $\eta_q(\lambda)$ is one-dimensional and convex, so its minimizer can also be found efficiently with bisection using the sign of the subgradient. Since these subproblems are independent across services, all inner solves can be carried out in parallel. The resulting procedure is summarized in Algorithm~\ref{alg:alpha_dual_decomp}.
\begin{algorithm}[t]
\caption{Dual decomposition with nested bisection for the reduced
$\alpha$-recourse problem.}
\label{alg:alpha_dual_decomp}
\hrule
\vspace{1mm}
\footnotesize
\begin{algorithmic}[1]
\STATEX \textbf{Algorithm 1:} Dual decomposition with nested bisection for $\alpha$-recourse \\
\STATEX \textbf{Input:} Service set $Q$, reduced budget $\tilde{R}$, functions $\{f_q,h_q\}_{q \in Q}$ \\
\STATEX \textbf{Output:} Service throttles $\alpha^\star$

\IF{$\sum_{q \in Q} h_q(1) \le \tilde{R}$} 
    \STATE Return $\alpha_q^\star = 1$ for all $q \in Q$
\ELSIF{$\sum_{q \in Q} h_q(\tau_q) > \tilde{R}$}
    \STATE Return $V_t(x_t^+)=+\infty$
\ENDIF

\STATE Set $\lambda_{\mathrm{lo}} \gets 0$ and choose $\lambda_{\mathrm{hi}}$ such that the induced solution is feasible

\WHILE{$\lambda_{\mathrm{hi}} - \lambda_{\mathrm{lo}} > \varepsilon_\lambda$}
    \STATE $\lambda \gets (\lambda_{\mathrm{lo}} + \lambda_{\mathrm{hi}})/2$
    \FORALL{$q \in Q$ \textbf{in parallel}}
        \STATE Compute
        \[
        \alpha_q^\star(\lambda) \in \arg\min_{\alpha_q \in [\tau_q,1]}
        f_q(\alpha_q) + \lambda h_q(\alpha_q)
        \]
        by bisection on $\alpha_q$
    \ENDFOR
    \IF{$\sum_{q \in Q} h_q(\alpha_q^\star(\lambda)) > \tilde{R}$}
        \STATE $\lambda_{\mathrm{lo}} \gets \lambda$
    \ELSE
        \STATE $\lambda_{\mathrm{hi}} \gets \lambda$
        \STATE Store current $\alpha^\star(\lambda)$
    \ENDIF
\ENDWHILE

\STATE Return last stored feasible solution
\end{algorithmic}
\vspace{1mm}
\hrule
\end{algorithm}
The dual variable $\lambda$ can be interpreted as an internal price signal for coordinating power allocation. Small values of $\lambda$ favor larger throttles to reduce service impact, while larger values penalize power consumption more strongly and favor smaller throttles. 
Algorithm \ref{alg:alpha_dual_decomp} is highly parallelizable and consists of only simple operations. The outer loop updates only a scalar dual variable, while the inner loop solves $|Q|$ independent one-dimensional convex problems via bisection. 

\section{A Structured Receding-Horizon Algorithm for Server Control}
\label{sec:ranked_prefix_controller}

We now present an approach for dynamically selecting server actions. 
The outer server-control problem (\ref{eq:reduced_dynamic_control_problem}) is a large, multi-stage discrete optimization problem. 
At a single time step, allowing arbitrary subsets of currently eligible servers to boot or shut down already yields an exponentially large action space. 
Over a horizon of length $H$, the number of possible server-action sequences grows combinatorially with both the number of eligible servers and the horizon length. 
This is further complicated by the delayed finite-state dynamics: a server action taken now may affect not only the realized power and capacity in the current time step, but also the future set of admissible actions and the future available capacity. 
Moreover, any candidate server trajectory must be evaluated by solving the $\alpha$-recourse problem across the horizon. 
Thus, even after separating the convex recourse layer, exact reduced MPC remains computationally impractical at datacenter scale.

To address this intractability, we develop a fast receding-horizon heuristic that optimizes the current server action using an $H$-step rollout and restricts the search over servers to ranked-prefix configurations. The fast $\alpha$-recourse algorithm is central to this approach, since each candidate configuration is scored based on its induced cost over the planning horizon.  

\subsection{Ranked-Prefix Receding-Horizon Approximation}

To obtain an efficient online server controller, we make two structural approximations. First, instead of optimizing server actions over the entire horizon, we use a one-step receding-horizon rollout. At each time step, we optimize only the current server action, rolling out the server dynamics assuming no further server actuation; server actions are then re-planned at every step. Second, we significantly compress the action space: instead of optimizing over arbitrary binary decisions for each server, we consider only \emph{how many} servers to boot or shut down. 

Let $x_t$ denote the current server state and let $u_t$ be a candidate current action. Applying $u_t$ yields the post-action state $x_t^+ = \phi(x_t,u_t)$.
We then simulate the server dynamics forward over the horizon assuming $\texttt{noop}$ server actions at times $t+1, \ldots, t+H-1$. Let $\{\hat x_{t+r}^+(u_t)\}_{r=0}^{H-1}$ denote the resulting server states. The candidate action $u_t$ is scored by the rollout objective
\begin{equation}
\mathcal{F}_{t}(u_{t})
=
\sum_{r=0}^{H-1} V_{t+r}\!\left(\hat x_{t+r}^+(u_t)\right),
\label{eq:server_rollout_score}
\end{equation}
where $V_{t+r}(\cdot)$ is the $\alpha$-recourse value function at time step $t+r$. 
The online controller applies the current action with the smallest rollout score, and replans at the next time step.

However, even this one-step server lookahead is still too computationally expensive if we search over arbitrary subsets of servers. 
We therefore impose the second aforementioned approximation by restricting the current action to one of three modes. 
At each time step, we only consider booting some number of eligible servers, shutting some number of eligible servers, or taking no action.
We do not consider arbitrary mixed boot/shutdown subsets.

For each type of action, we rank the eligible servers using a cheap static proxy for their marginal value. 
Define the weighted nominal capacity hosted on server $s$ as $\kappa_s = \sum_{w \in W_s} \theta_{q(w)} \, g_w^{\textrm{above}}(1),
$
which measures the weighted (by service priority) capacity contributed by server $s$ when its hosted workers operate at full throttle. 
We then rank servers for shutdown and boot using the ratio between their weighted capacity and baseline power draw: 
\begin{equation}
r_s^{\mathrm{shut}} = \frac{\mathrm{PUE}\, b_s}{\kappa_s},
\label{eq:shutdown_rank_score}
\end{equation}
\begin{equation}
r_s^{\mathrm{boot}} = \frac{\kappa_s}{\mathrm{PUE}\, b_s}.
\label{eq:boot_rank_score}
\end{equation}

Let $s_{(1)}^{\mathrm{shut}}, s_{(2)}^{\mathrm{shut}}, \dots$ and 
$s_{(1)}^{\mathrm{boot}}, s_{(2)}^{\mathrm{boot}}, \dots$ 
denote the eligible shutdown and boot servers ordered from best to worst under these rankings. 
Rather than searching over arbitrary subsets, we restrict attention to \emph{ranked-prefix} actions of the form ``act on the top $k$ servers'' for some depth $k$. 
Thus, for each mode, the current action is parameterized by a single integer $k$.
For each mode $a \in \{\texttt{BOOT},\texttt{SHUT}\}$, let $u_t^a(k)$ denote the action that applies mode $a$ to the top $k$ eligible servers under the corresponding ranking and applies \texttt{noop} to all others. Its rollout score is $\mathcal{F}_t^a(k) = \mathcal{F}_t(u_t^a(k))$ where $k=0$ denotes no action. The controller selects the mode and depth with the smallest rollout score. 

The ranking step is deliberately modular. 
It produces an ordered candidate list cheaply, while the final decision is made using the dynamic rollout score in \eqref{eq:server_rollout_score}, which accounts for delayed server dynamics and the downstream $\alpha$-recourse. 

\subsection{Geometric Depth Search}

Under the ranked-prefix restriction, each action mode $a \in \{\texttt{BOOT},\texttt{SHUT}\}$ induces a one-dimensional search problem over $k \in \{0,\ldots,N_t^a\}$ where $N_t^a$ is the number of eligible servers for mode $a$.
In general, we cannot assume the rollout score $\mathcal{F}_t^a(k)$ to be monotone or unimodal, because the throttle recourse is reoptimized for every prefix. 
For shutdowns, increasing $k$ adds a less favorable server to the shutdown set and removes additional capacity, but it may also create power headroom that permits higher throttles on the remaining workers. 
Booting produces a corresponding tradeoff between restoring capacity and consuming additional power. 

To explore the action space with few rollout evaluations,  we use a coarse-to-fine geometric search over $k$. For each mode $a$, we generate a geometric \emph{probe sequence} $\{k_j^a\}_{j\geq 0}$ with common ratio $\rho>1$, either in the forward direction (for shutdowns)
\begin{equation}
k_{j+1}^a = \min\{N_t^a,\lfloor \rho k_j^a \rfloor\},
\end{equation}
or in the reverse direction (for boots)
\begin{equation}
k_{j+1}^a = \max\{1,\lfloor k_j^a/\rho \rfloor\}.
\end{equation}
The search spans small, medium, and large reconfiguration depths using only logarithmically many rollout evaluations.

During this probing phase, we maintain the best finite score seen so far and terminate when the score has failed to improve for $\nu$ consecutive probes or the eligible list is exhausted, where $\nu$ is a small fixed parameter. Let $k^{a,\star}$ denote the best coarse depth found, and let $k_-^a$ and $k_+^a$ denote the adjacent coarse probes bracketing it. We then refine locally by evaluating uniformly spaced depths in the interval $[k_-^a,k_+^a]$. For $m$ local refinements, we evaluate the additional depths

\begin{equation}
\hspace{-1.0em}\mathcal{R}_t^a = \left\{ \textrm{round} \left( k_{-}^a + \frac{\ell (k_+^a - k_-^a)}{m - 1} \right) \mid \ell = 0, \dots, m-1 \right\}
\end{equation}
after removing endpoints. The selected depth is the one with minimum rollout score among all evaluations.

\section{Experimental Setup}
We evaluate the proposed controller on realistic-scale datacenter instances to assess closed-loop performance and scalability under representative power budget trajectories. 
We additionally study small synthetic instances, where the optimal solution can be computed exactly to compare against our proposed method. 
Across these experiments, we compare against several baseline methods, including throttling only and server only controllers, as well as centralized optimization baselines. 
\subsection{Problem Instances}

\subsubsection{Realistic Instance}
We generate a full datacenter scale instance using the datacenter simulation tools developed in \cite{lin_slasher_2026}. This allows us to generate synthetic workloads calibrated to production statistics. 
The instance comprises a datacenter with $|S|=15360$ servers, $|W|=207674$ workers, and $|Q|=1422$ services. Services are organized into three priority tiers with $\theta \in \{1, 10, 1000\}$ to reflect varying levels of criticality. 
Each service is replicated across a subset of servers, producing a dense overlapping placement that is typical in production environments. 
The worker power model is convex quadratic with $a_q=5.5$ and $b_q=2.75$. The capacity model is concave with all services having $\gamma_q=0.7$, and the $\text{PUE}=1.2$. 
Impact functions are approximated by sampling $N_q=20$ samples from each service's load distribution, which is simulated using the generated workload trace data.
Further details about the generation process can be found in Appendix \ref{appendix:experiment_details_realistic}. We set the transition delay to $K=3$ with a planning horizon of $H=6$ and run for an episode length of $T=30$. In a realistic instance, the length of each time step is $20$ seconds; thus, the controller must produce a decision within $20$ seconds to be implementable. 

\subsubsection{Synthetic Instances}
For validation against exact methods, we additionally generate small datacenter instances with $|S| \in \{3,\dots,7\}$ servers with a similar worker distribution as the realistic instance and $\sim 3|S|$ services with random overlapping placements across three priority tiers. We set the transition delay to $K=2$ with a planning horizon of $H=4$ and run for an episode length of $T=10$. We generate instances using several random seeds to vary service placement and demand realizations. See Appendix~\ref{appendix:experiment_details_realistic} for further details on instance generation.  

\subsection{Controllers and Baselines}

We evaluate six controllers across three different categories to characterize the design space. 
Formal mathematical formulations for all the non-real-time controllers can be found in Appendix \ref{appendix:controller_formulations}.

The following three methods are real-time controllers.

\subsubsection{AlphaOnly}
All servers remain $\texttt{ON}$. Only the per-service throttle $\alpha_q$ is optimized at each time step.

\subsubsection{ServerOnly}
Uses the \emph{RankedPrefix} server planner with all throttles fixed at $\alpha=1$. There is no convex recourse.

\subsubsection{RankedPrefix}
Our method described in Section \ref{sec:ranked_prefix_controller}.

The next two methods are one-step optimization baselines that solve a single step optimization with the same forward rollout as the heuristic methods. 

\subsubsection{GurobiFixedAlpha}
Solves a MILP over server actions with all workers set to $\alpha=1$.

\subsubsection{GurobiMICP}
Jointly optimizes over server actions and per-service throttles $\alpha_q$ by solving a mixed integer convex program (MICP).

Finally, we have an exact oracle method against which we can validate our real-time controllers on small instances. 

\subsubsection{Dynamic Programming (DP) MPC}
The oracle solves the server-only problem using dynamic programming over a time expanded server state graph. The $\alpha$-recourse costs are the node transition costs in the graph. 

\subsection{Scenarios}

Each experiment is characterized by a time-varying power budget trajectory $\{\hat{P}^\text{max}_t\}_{t=0}^{T-1}$. We study three trajectories.

\begin{itemize}
    \item \textbf{Step}: An instantaneous drop at $t=T/4$ that instantaneously recovers at $t=3T/4$.
    \item \textbf{Ramp}: Descends linearly at $t=2T/5$, holds at reduced power until $3T/5$, and recovers linearly at $4T/5$.
    \item \textbf{Oscillate}: A cosine shaped power signal over the first $4T/5$ slots that recovers for the remaining $T/5$ slots.
\end{itemize}

\section{Results}
We evaluate the proposed controller on its closed-loop performance on realistic large-scale instances and computational scalability; we also evaluate its performance relative to exact dynamic programming on small instances.

\subsection{Closed-Loop Performance on Realistic Sizes}


Table~\ref{tab:realtime_summary} compares the three real-time controllers under step, ramp, and oscillatory power-cap trajectories. Across all three, \emph{RankedPrefix} achieves the lowest episode cost while incurring no budget violations and remaining within the 20s planning constraint. 

\begin{figure*}[!t]
    \centering
    \includegraphics[width=0.89\textwidth,trim={0.2cm 0.1cm 0.2cm 0.1cm},clip]{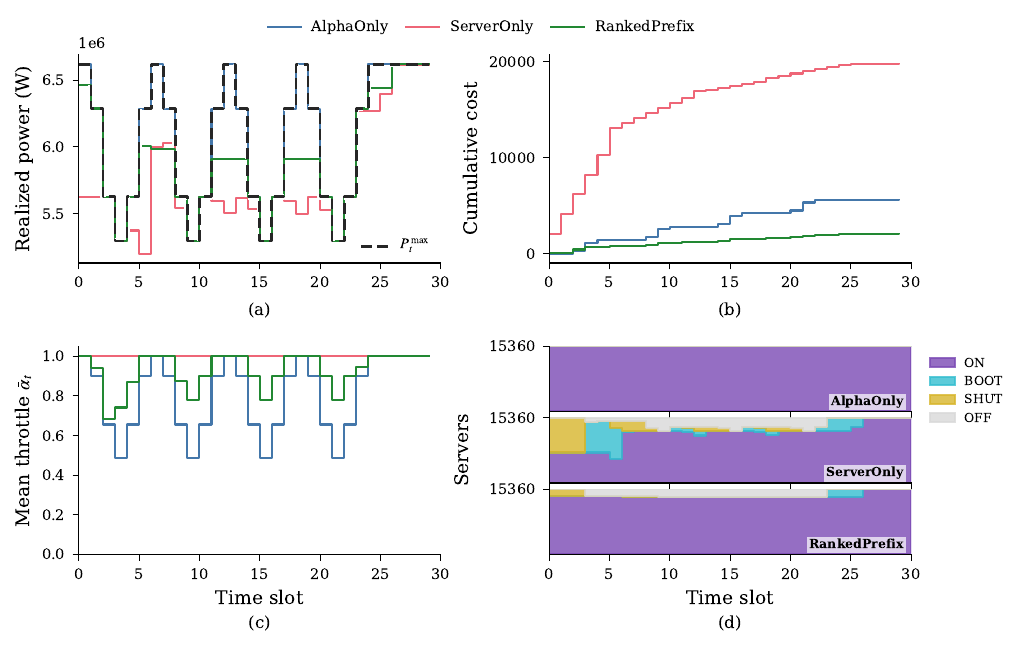}
    \caption{Real-time controller performance on oscillatory power trajectory. (a) \emph{RankedPrefix} satisfies the power budget (black dashed line). (b) Final cumulative cost is $63.7\%$ lower than \emph{AlphaOnly} and $89.7\%$ lower than \emph{ServerOnly}. (c)--(d) Coordinating both mechanisms avoids aggressive throttling and excessive server shutdowns}
    \label{fig:oscillate20}
    \vspace{-1.5em}
\end{figure*}

\begin{table}[t]
\caption{Controller performance across realistic scenarios.}
\label{tab:realtime_summary}
\centering
\footnotesize
\setlength{\tabcolsep}{3.5pt}
\renewcommand{\arraystretch}{1.08}
\begin{tabular}{|c|c|c|c|c|}
\hline
\textbf{Scenario} & \textbf{Controller} & \textbf{Episode Cost} & \textbf{\# Viol.} & \textbf{Max Plan Time (s)} \\
\hline
\multirow{3}{*}{Step}
    & AlphaOnly    & 12,529.6 & \textbf{0} & \textbf{0.0018} \\
    & ServerOnly   & 4,241.8  & \textbf{0} & 1.55 \\
    & RankedPrefix & \textbf{2,761.0} & \textbf{0} & 10.8 \\
\hline
\multirow{3}{*}{Ramp}
    & AlphaOnly    & 152,298.9 & 11 & \textbf{0.0019} \\
    & ServerOnly   & 8,502.0   & \textbf{0} & 1.51 \\
    & RankedPrefix & \textbf{6,683.2} & \textbf{0} & 18.4 \\
\hline
\multirow{3}{*}{Oscillate}
    & AlphaOnly    & 5,620.5  & 8 & \textbf{0.0017} \\
    & ServerOnly   & 19,733.3 & \textbf{0} & 1.58 \\
    & RankedPrefix & \textbf{2,037.6} & \textbf{0} & 14.8 \\
\hline
\end{tabular}
\end{table}

Fig.~\ref{fig:oscillate20} illustrates the qualitative behavior of various real-time controllers on an oscillatory power trajectory. \emph{AlphaOnly} relies on aggressive throttling during low-cap intervals, which comes at the cost of service capacity. With no fast recourse, \emph{ServerOnly} instead shuts down many servers, sacrificing capacity in order to wait for delayed power savings. \emph{RankedPrefix} strikes a balance between both mechanisms by combining fewer shutdowns with moderate throttling to achieve the lowest cumulative cost. 




\subsection{Computational Scalability and Empirical Optimality}

\begin{figure*}[t]
    \centering
    \includegraphics[width=0.89\textwidth]{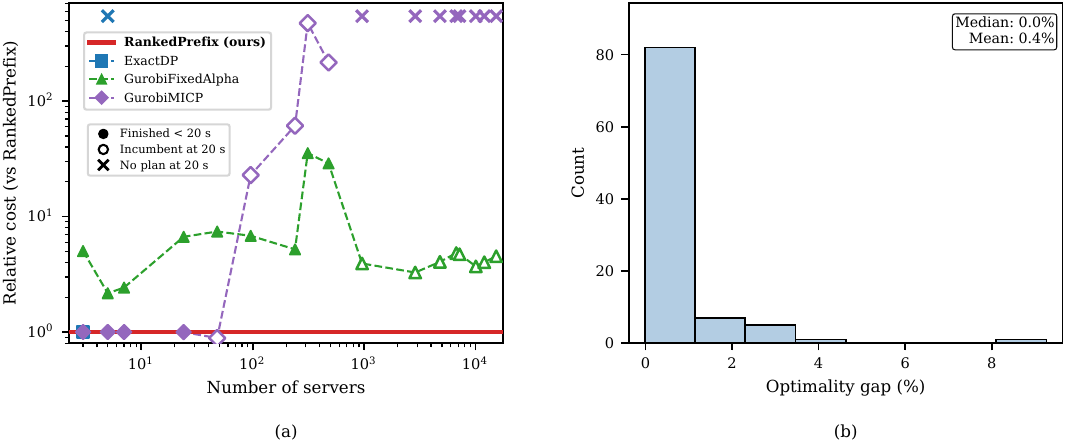}
    \caption{Scalability and comparison with exact MPC. (a) \emph{RankedPrefix} finds low-cost plans at realistic scale under a 20s per-step planning limit. Episode costs are normalized by \emph{RankedPrefix}. Filled markers indicate completion within the limit, open markers indicate a feasible incumbent within the limit, and $\times$ indicates no feasible plan found. (b) On 100 small instances, \emph{RankedPrefix} closely matches exact MPC, with median and mean episode cost gaps of $0.0\%$ and $0.4\%$, respectively. }
    \label{fig:scaling_panel}
    \vspace{-1.5em}
\end{figure*}



Fig.~\ref{fig:scaling_panel}a reports episode cost relative to \emph{RankedPrefix} under a $20$s per-step solve time limit. \emph{ExactDP} is tractable only for very small instances, while \emph{GurobiMICP} is competitive at small sizes but returns poor or no incumbents as the problem grows. 
\emph{GurobiFixedAlpha} scales further but incurs higher cost relative to \emph{RankedPrefix} as it omits throttling recourse. Our \emph{RankedPrefix} controller benefits from coordinating server reconfiguration with throttling recourse to scale beyond the centralized optimization baselines. Variation across server sizes reflects the independently generated workload mixes. 



Fig.~\ref{fig:scaling_panel}c evaluates optimality on 100 random small instances with varying topologies, service placements, and power-cap trajectories where dynamic programming is tractable. The distribution of optimality gaps is concentrated near 0, showing that \emph{RankedPrefix} is usually either optimal or near-optimal on small instances where exact comparison is possible. 




\section{Conclusion}
Datacenter power modulation is fundamentally a multi-timescale control problem. Operational power flexibility is distributed across many interacting layers and actuated by different control mechanisms. In this work, we have considered fast worker throttling and slower server reconfiguration. Our results show that treating either mechanism in isolation leads to poor closed-loop behavior under time-varying power limits relative to a controller that combines both. 

We propose a receding-horizon controller that coordinates these mechanisms while remaining tractable at datacenter scale. For fixed server states, the throttling problem becomes a service-level convex recourse problem that can be solved efficiently via dual decomposition. We develop a ranked-prefix approximation, where structured server actions are evaluated using the downstream recourse value over the forecast horizon. This yields an online controller for modulating datacenter power without having to solve a mixed-integer dynamic optimization, which is intractable at scale. 

Across realistic power-cap trajectories, our proposed controller reduces cost relative to throttling-only and server-only baselines while satisfying power limits in real time, with planning times below the 20s control interval. Scaling experiments show that our method is tractable at realistic datacenter sizes, and evaluation across a variety of smaller instances indicates its performance typically approaches that of exact MPC.

There are many possible avenues for future work: for example, there remains an opportunity to extend our proposed framework to richer datacenter models, including heterogeneous accelerator workloads, different service classes, and more detailed performance models, as part of a broader approach to datacenter operational power flexibility. 

\appendices

\section{Baseline and Oracle Formulations}
\label{appendix:controller_formulations}

\subsection{Exact MPC via Dynamic Programming}

For small instances, we can solve the reduced finite horizon server control problem in (\ref{eq:reduced_dynamic_control_problem}) exactly via dynamic programming. Let $\mathcal{X}$ denote the finite set of feasible server state vectors, and let $\mathcal{U}(x)$ denote the set of admissible server action vectors from a state $x \in \mathcal{X}$. For a fixed solve over the horizon at time $t$, we define $J^\star_r(x)$ as the optimal remaining cost from a pre-action server state $x$ at a horizon step $r$. This corresponds to the physical time $t+r$. We define the terminal condition as $J_H^\star(x) = 0$ for all $x \in \mathcal{X}$. 

The Bellman recursion for all horizon steps $r=0,\dots,H-1$ is

\begin{equation}
J_r^\star(x)
=
\underset{u \in \mathcal{U}(x)}{\textrm{minimize}} \quad 
V_{t+r}(\phi(x,u))
+
J_{r+1}^\star(\psi(\phi(x,u)))
,
\label{eq:dp_bellman}
\end{equation}
 where the stage cost $V_{t+r}(\phi(x,u))$ is the $\alpha$-recourse value. 

\subsection{Optimal One-Step Fixed Throttle Controller}
\label{appendix:gurobi_fixed_alpha}

We describe the optimal one-step fixed throttle controller \emph{GurobiFixedAlpha}. 
Let $x_t$ denote the current server state, and let $\mathcal{U}(x_{s,t})$ denote the admissible actions for server $s$. For each server $s$ and admissible action $a \in \mathcal{U}(x_{s,t})$, we introduce a binary variable $y_{s,a}$ satisfying $\sum_{a \in \mathcal{U}(x_{s,t})} y_{s,a} = 1$ for all $s \in \mathcal{S}$.

For each action $a$ and rollout step $r=0, \ldots, H-1$, precompute the induced indicators $\bar{z}^P_{s,r}(a)$ and $\bar{z}^C_{s,r}(a)$. The rollout indicators and capacity are

\begin{equation}
\begin{aligned}
z^{P}_{s,t+r}
&=
\sum_{a\in\mathcal{U}(x_{s,t})}
\bar z^{P}_{s,r}(a)\,y_{s,a},
\\
z^{C}_{s,t+r}
&=
\sum_{a\in\mathcal{U}(x_{s,t})}
\bar z^{C}_{s,r}(a)\,y_{s,a},
\end{aligned}
\label{eq:rollout_indicators}
\end{equation}

\begin{equation}
C_{q,t+r}
=
\sum_{w \in W_q} z^{C}_{s(w),t+r}\, g_w^{\textrm{above}}(1).
\label{eq:fixedalpha_capacity}
\end{equation}

Let $P^{\textrm{ovh}}_{t+r} = B + \sum_{s \in \mathcal{S}}
\bigl[
z^{P}_{s,t+r} b_s
+
\left(z^{P}_{s,t+r}-z^{C}_{s,t+r}\right) P_{\mathrm{tr},s}
\bigr]$ denote the server state overhead power. The total datacenter power is then $P^{\mathrm{tot}}_{t+r} = \mathrm{PUE}\left(P^{\textrm{ovh}}_{t+r} + \sum_{s \in \mathcal{S}} z^{C}_{s,t+r} \sum_{w \in \mathcal{W}_s} p_w^{\textrm{above}}(1)\right)$.

Introducing epigraph variables $\xi_{q,j,r}$ for the sampled expected-shortfall terms yields

\begin{equation}
\begin{aligned}
\underset{\{y_{s,a}\}, \{\xi_{q,j,r}\}}{\textrm{minimize}} \quad &
\sum_{r=0}^{H-1} \bigg[\sum_{q\in Q} \frac{\theta_q}{N\bar{l}_q} \sum_{j=1}^{N} \xi_{q,j,r} +  \varepsilon D(\mathbf{1}, z^{P}_{t+r}, z^{C}_{t+r}) \bigg] \\
\text{s.t.}\quad
& \sum_{a \in \mathcal{U}(x_{s,t})} y_{s,a}=1, \quad \forall s \in \mathcal{S}, \\
& P_{t+r}^{\mathrm{tot}} \leq \hat P_{t+r}^{\max}, \quad r=0, \dots, H-1, \\
& \xi_{q,j,r} \geq 0, \quad  \xi_{q,j,r} \geq l^{(j)}_{q,t+r} - C_{q,t+r}, \\
& \hspace{4em} \forall q, j=1,\dots,N, r=0,\dots, H-1, \\
\end{aligned}
\label{eq:gurobi_fixed_alpha}
\end{equation}

Each time we solve (\ref{eq:gurobi_fixed_alpha}), we apply only the current optimal server action and then re-plan at the next time step. 

\subsection{Optimal One-Step Mixed Integer Convex Controller}
We describe the optimal one-step mixed integer convex controller \emph{GurobiMICP}. We use the same current action binaries $y_{s,a}$ and induced rollout indicators $z^P_{s,t+r}$ and $z^C_{s,t+r}$ defined in Appendix~\ref{appendix:gurobi_fixed_alpha}. Unlike \emph{GurobiFixedAlpha}, we now also optimize the service-level throttles over the rollout.

Using the variable reduction from Section~\ref{sec:variable_reduction}, we maintain one throttle $\alpha_{q,t+r} \in [\tau_q, 1]$ for each service $q$ and rollout step $r=0,\dots,H-1$. For each service $q$, server $s$, and rollout step $r$, let $n_{q,s}$ denote the number of workers of service $q$ hosted on server $s$. Because we must also optimize over throttles in this case, we introduce an auxiliary variable $u_{q,s,t+r}$ which represents the realized \emph{above-threshold} throttle on server $s$. We will want to enforce that $u_{q,s,t+r}=(\alpha_{q,t+r}-\tau_q)z^C_{s,t+r}$. We also add auxiliary variables $p_{q,s,t+r}$ and $c_{q,s,t+r}$ for the worker power and worker capacity contributed by service $q$ on server $s$ at time $t+r$. 

The realized capacity of service $q$ at time $t+r$ is given by $C_{q,t+r} = \sum_{s \in \mathcal{S}} c_{q,s,t+r}$.
The total power at time $t+r$ is then $P^{\mathrm{tot}}_{t+r} = \mathrm{PUE}\left(P^{\textrm{ovh}}_{t+r} + \sum_{q \in \mathcal{Q}} \sum_{s \in \mathcal{S}} p_{q,s,t+r}\right)$.

Using the same linearization of the expected shortfall objective, we can write the one-step mixed integer convex controller as follows. Let $\mathcal{V}$ denote the collection of decision variables $\{y_{s,a}\},\{\alpha_{q,t+r}\},\{u_{q,s,t+r}\},
\{p_{q,s,t+r}\},\{c_{q,s,t+r}\},\{\xi_{q,j,r}\}$.
\begingroup
\allowdisplaybreaks[4]
\begin{align}
\underset{\mathcal{V}}{\textrm{minimize}} \quad &
\sum_{r=0}^{H-1}
\bigg[
\sum_{q\in\mathcal{Q}}
\frac{\theta_q}{N\bar{l}_q}
\sum_{j=1}^{N}\xi_{q,j,r}
+
\epsilon D(\alpha_{t+r},z^P_{t+r},z^C_{t+r})
\bigg]
\notag\\
\text{s.t.}\quad
&
\sum_{a\in\mathcal{U}_s(x_{s,t})}y_{s,a}=1,
\qquad \forall s\in\mathcal{S},
\notag\\
&
\tau_q\leq\alpha_{q,t+r}\leq1,
\qquad
\forall q\in\mathcal{Q},\ r=0,\ldots,H-1,
\notag\\
&
0\leq u_{q,s,t+r}
\leq(1-\tau_q)z^C_{s,t+r},
\notag\\
&
u_{q,s,t+r}
\leq\alpha_{q,t+r}-\tau_q,
\notag\\
&
u_{q,s,t+r}
\geq
\alpha_{q,t+r}-\tau_q
-(1-\tau_q)(1-z^C_{s,t+r}),
\notag\\
&
p_{q,s,t+r}
\geq
n_{q,s}\,
p_q^{\textrm{above}}(\tau_q+u_{q,s,t+r}),
\notag\\
&
c_{q,s,t+r}
\leq
n_{q,s}\,
g_q^{\textrm{above}}(\tau_q+u_{q,s,t+r}),
\notag\\
&
\hspace{2em}
\forall q\in\mathcal{Q},\
s\in\mathcal{S},\
r=0,\ldots,H-1,
\notag\\
&
\xi_{q,j,r}\geq0,
\qquad
\xi_{q,j,r}
\geq l_{q,t+r}^{(j)}-C_{q,t+r},
\notag\\
&
\hspace{2em}
\forall q\in\mathcal{Q},\
j=1,\ldots,N,\
r=0,\ldots,H-1,
\notag\\
&
P_{t+r}^{\mathrm{tot}}
\leq\hat P_{t+r}^{\max},
\qquad r=0,\ldots,H-1.
\label{eq:gurobi_micp}
\end{align}
\endgroup
where constraints (3)-(5) allow us to linearize the bilinear constraint $u_{q,s,t+r}=(\alpha_{q,t+r}-\tau_q)z^C_{s,t+r}$, and constraints (6) and (7) constrain the power and capacity. As with \emph{GurobiFixedAlpha}, each time we solve (\ref{eq:gurobi_micp}), we apply on the current optimal action and then re-plan at the next time step.

\section{Details on Instance Generation}
\label{appendix:experiment_details_realistic}

The simulator from \cite{lin_slasher_2026} produces a synthetic datacenter instance through a hierarchical, power-driven procedure. Given a power budget and a catalog of server SKUs, servers are iteratively sampled and placed onto a tile-based floor plan until the allocated power is met. Service are sampled from a heavy-tailed priority distribution, and each service spawns workers until a target core utilization is reached. Workers are assigned to servers using a configurable allocator that respects per-server core capacity. Per-worker CPU utilization time series are synthesized from a flat baseline, daily and weekly seasonalities, and Gaussian noise, reproducing the patterns characteristic of production workloads. Each worker's time series is aggregated into a per-service histogram over utilization (load), used to sample loads for the impact function approximation.

Small synthetic instances used for validation against exact methods are generated procedurally rather than from traces. For a target server count $|S|$, we instantiate a single datacenter populated by servers of fixed base power and core count. Services are sampled across three priority tiers, with the number of services in each tier scaling with $|S|$ and biased toward the low-priority tier to mirror the heavy-tailed composition of production workloads. Each service is placed on a uniformly random subset of servers whose size is drawn from a tier-dependent range, producing the dense overlapping placement characteristic of production deployments. Per-service demand samples are drawn from a Gaussian distribution, clipped at zero. Random seeds vary both the service placement and the demand realizations across instances. 

\bibliographystyle{IEEEtran}
\bibliography{references_2}

@techreport{larson_utility_2024,
    author      = {{EPRI}},
    title       = {Utility Experiences and Trends Regarding Data Centers: 2024 Survey},
    institution = {EPRI},
    address     = {Palo Alto, CA, USA},
    number      = {3002030643},
    month       = sep,
    year        = {2024},
}

@inproceedings{li_thunderbolt_2020,
	title = {Thunderbolt: {Throughput}-{Optimized}, {Quality}-of-{Service}-{Aware} {Power} {Capping} at {Scale}},
	isbn = {ISBN 978-1-939133-19-9},
	url = {},
	booktitle = {OSDI},
	author = {Li, Shaohong and Wang, Xi and Zhang, Xiao and Kontorinis, Vasileios and Kodakara, Sreekumar and Lo, David and Ranganathan, Parthasarathy},
	month = nov,
	year = {2020},
}

@inproceedings{li_scalable_2019,
	title = {A {Scalable} {Priority}-{Aware} {Approach} to {Managing} {Data} {Center} {Server} {Power}},
	url = {},
	doi = {10.1109/HPCA.2019.00067},
	urldate = {2026-05-13},
	booktitle = {HPCA},
	author = {Li, Yang and Lefurgy, Charles R. and Rajamani, Karthick and Allen-Ware, Malcolm S. and Silva, Guillermo J. and Heimsoth, Daniel D. and Ghose, Saugata and Mutlu, Onur},
	month = feb,
	year = {2019},
}

@inproceedings{urgaonkar_dynamic_2005,
	title = {Dynamic {Provisioning} of {Multi}-tier {Internet} {Applications}},
	url = {},
	doi = {10.1109/ICAC.2005.27},
	urldate = {2026-05-08},
	booktitle = {ICAC},
	author = {Urgaonkar, B. and Shenoy, P. and Chandra, A. and Goyal, P.},
	month = jun,
	year = {2005},
	pages = {},
}

@article{gandhi_autoscale_2012,
	title = {{AutoScale}: {Dynamic}, {Robust} {Capacity} {Management} for {Multi}-{Tier} {Data} {Centers}},
	volume = {30},
	issn = {0734-2071, 1557-7333},
	shorttitle = {{AutoScale}},
	url = {},
	doi = {10.1145/2382553.2382556},
	language = {en},
	number = {4},
	urldate = {2026-05-08},
	journal = {ACM Transactions on Computer Systems},
	author = {Gandhi, Anshul and Harchol-Balter, Mor and Raghunathan, Ram and Kozuch, Michael A.},
	month = {},
	year = {2012},
	pages = {},
}

@inproceedings{albers_algorithms_2021,
	title = {Algorithms for {Right}-{Sizing} {Heterogeneous} {Data} {Centers}},
	isbn = {9781450380706},
	url = {},
	doi = {10.1145/3409964.3461789},
	language = {en},
	urldate = {2026-05-08},
	booktitle = {SPAA},
	publisher = {ACM},
	author = {Albers, Susanne and Quedenfeld, Jens},
	month = jul,
	year = {2021},
}

@article{lin_dynamic_2013,
	title = {Dynamic {Right}-{Sizing} for {Power}-{Proportional} {Data} {Centers}},
	volume = {21},
	issn = {1558-2566},
	url = {},
	doi = {10.1109/TNET.2012.2226216},
	number = {5},
	urldate = {2026-05-08},
	journal = {IEEE/ACM Transactions on Networking},
	author = {Lin, Minghong and Wierman, Adam and Andrew, Lachlan L. H. and Thereska, Eno},
	month = oct,
	year = {2013},
	pages = {1378--1391},
}

@inproceedings{verma_large-scale_2015,
	title = {Large-scale cluster management at {Google} with {Borg}},
	isbn = {9781450332385},
	url = {},
	doi = {10.1145/2741948.2741964},
	language = {en},
	urldate = {2026-05-08},
	booktitle = {EuroSys},
	publisher = {ACM},
	author = {Verma, Abhishek and Pedrosa, Luis and Korupolu, Madhukar and Oppenheimer, David and Tune, Eric and Wilkes, John},
	month = apr,
	year = {2015},
}

@article{marcucci_warm_2021,
	title = {Warm {Start} of {Mixed}-{Integer} {Programs} for {Model} {Predictive} {Control} of {Hybrid} {Systems}},
	volume = {66},
	issn = {1558-2523},
	url = {},
	doi = {10.1109/TAC.2020.3007688},
	number = {6},
	urldate = {2026-04-30},
	journal = {IEEE Transactions on Automatic Control},
	author = {Marcucci, Tobia and Tedrake, Russ},
	month = jun,
	year = {2021},
	pages = {2433--2448},
}

@inproceedings{kanev_tradeoffs_2014,
	title = {Tradeoffs between power management and tail latency in warehouse-scale applications},
	url = {},
	doi = {10.1109/IISWC.2014.6983037},
	urldate = {2026-05-07},
	booktitle = {IISWC},
	author = {Kanev, Svilen and Hazelwood, Kim and Wei, Gu-Yeon and Brooks, David},
	month = oct,
	year = {2014},
}

@inproceedings{qiu_revisiting_2025,
	title = {Revisiting {CPU} {Performance} {Scaling} for {Energy}-{Efficient} {Packet} {Processing} {Applications}},
	isbn = {9798400711251},
	url = {},
	doi = {10.1145/3679240.3734588},
	language = {en},
	urldate = {2026-05-07},
	booktitle = {ACM e-Energy},
	publisher = {ACM},
	author = {Qiu, Biqing and Chen, Yixi and Khooi, Xin Zhe and Song, Cha Hwan and Chan, Mun Choon},
	month = jun,
	year = {2025},
}

@techreport{norris_rethinking_2025,
	title = {Rethinking {Load} {Growth}: {Assessing} the {Potential} for {Integration} of {Large} {Flexible} {Loads} in {US} {Power} {Systems}},
	institution = {Nicholas Institute for Energy, Environment, and Sustainability},
	author = {Norris, Tyler H. and Profeta, Tim and Patino-Echeverri, Dalia and Cowie-Haskell, Adam},
	year = {2025},
}

@book{borrelli_predictive_2017,
	title = {Predictive {Control} for {Linear} and {Hybrid} {Systems}},
	publisher = {Cambridge University Press},
	author = {Borrelli, Francesco and Bemporad, Alberto and Morari, Manfred},
	year = {2017},
}

@article{jahanshahi_powermorph_2022,
	title = {{PowerMorph}: {QoS}-{Aware} {Server} {Power} {Reshaping} for {Data} {Center} {Regulation} {Service}},
	volume = {19},
	issn = {1544-3566, 1544-3973},
	shorttitle = {{PowerMorph}},
	url = {},
	doi = {10.1145/3524129},
	language = {en},
	number = {3},
	urldate = {2026-04-09},
	journal = {ACM Transactions on Architecture and Code Optimization},
	author = {Jahanshahi, Ali and Yu, Nanpeng and Wong, Daniel},
	month = sep,
	year = {2022},
}

@techreport{blanford_powering_2026,
	title = {Powering {Intelligence} 2026: {Updated} {Scenarios} of {U}.{S}. {Data} {Center} {Electricity} {Use} and {Power} {Strategies}},
	url = {},
	number = {3002034696},
	urldate = {2026-04-06},
	institution = {EPRI},
	author = {Blanford, Geoff and Wilson, Tom and Bistline, John and Johnson, Nils},
	month = feb,
	year = {2026},
}

@inproceedings{savasci_pads_2024,
	title = {{PADS}: {Power} {Budgeting} with {Diagonal} {Scaling} for {Performance}-{Aware} {Cloud} {Workloads}},
	copyright = {https://doi.org/10.15223/policy-029},
	isbn = {9798331507862},
	shorttitle = {{PADS}},
	url = {},
	doi = {10.1109/IGSC64514.2024.00012},
	urldate = {2026-04-09},
	booktitle = {IGSC},
	publisher = {IEEE},
	author = {Savasci, Mehmet and Souza, Abel and Irwin, David and Ali-Eldin, Ahmed and Shenoy, Prashant},
	month = nov,
	year = {2024},
	pages = {},
}

@article{chiang_layering_2007,
	title = {Layering as {Optimization} {Decomposition}: {A} {Mathematical} {Theory} of {Network} {Architectures}},
	volume = {95},
	issn = {1558-2256},
	shorttitle = {Layering as {Optimization} {Decomposition}},
	url = {},
	doi = {10.1109/JPROC.2006.887322},
	number = {1},
	urldate = {2026-04-30},
	journal = {Proceedings of the IEEE},
	author = {Chiang, Mung and Low, Steven H. and Calderbank, A. Robert and Doyle, John C.},
	month = jan,
	year = {2007},
}

@article{wang_frequency_2019,
	title = {Frequency regulation service provision in data center with computational flexibility},
	volume = {251},
	issn = {03062619},
	url = {},
	doi = {10.1016/j.apenergy.2019.05.107},
	language = {en},
	urldate = {2026-04-30},
	journal = {Applied Energy},
	author = {Wang, Wei and Abdolrashidi, Amirali and Yu, Nanpeng and Wong, Daniel},
	month = oct,
	year = {2019},
	pages = {113304},
}

@article{hager_exploring_2016,
	title = {Exploring performance and power properties of modern multi‐core chips via simple machine models},
	volume = {28},
	copyright = {http://onlinelibrary.wiley.com/termsAndConditions\#vor},
	issn = {1532-0626, 1532-0634},
	url = {},
	doi = {10.1002/cpe.3180},
	language = {en},
	number = {2},
	urldate = {2026-04-22},
	journal = {Concurrency and Computation: Practice and Experience},
	author = {Hager, Georg and Treibig, Jan and Habich, Johannes and Wellein, Gerhard},
	month = feb,
	year = {2016},
	pages = {189--210},
}

@inproceedings{wu_dynamo_2016,
	title = {Dynamo: {Facebook}'s {Data} {Center}-{Wide} {Power} {Management} {System}},
	isbn = {9781467389471},
	shorttitle = {Dynamo},
	url = {},
	doi = {10.1109/ISCA.2016.48},
	urldate = {2026-04-09},
	booktitle = {ISCA},
	publisher = {IEEE},
	author = {Wu, Qiang and Deng, Qingyuan and Ganesh, Lakshmi and Hsu, Chang-Hong and Jin, Yun and Kumar, Sanjeev and Li, Bin and Meza, Justin and Song, Yee Jiun},
	month = jun,
	year = {2016},
}

@inproceedings{ghamkhari_data_2012,
	title = {Data centers to offer ancillary services},
	url = {},
	doi = {10.1109/SmartGridComm.2012.6486023},
	urldate = {2026-04-30},
	booktitle = {SmartGridComm},
	author = {Ghamkhari, Mahdi and Mohsenian-Rad, Hamed},
	month = nov,
	year = {2012},
}

@inproceedings{sakalkar_data_2020,
	title = {Data {Center} {Power} {Oversubscription} with a {Medium} {Voltage} {Power} {Plane} and {Priority}-{Aware} {Capping}},
	isbn = {9781450371025},
	url = {},
	doi = {10.1145/3373376.3378533},
	language = {en},
	urldate = {2026-04-09},
	booktitle = {ASPLOS},
	publisher = {ACM},
	author = {Sakalkar, Varun and Kontorinis, Vasileios and Landhuis, David and Li, Shaohong and De Ronde, Darren and Blooming, Thomas and Ramesh, Anand and Kennedy, James and Malone, Christopher and Clidaras, Jimmy and Ranganathan, Parthasarathy},
	month = {},
	year = {2020},
	pages = {},
}

@inproceedings{patel_characterizing_2024,
	title = {Characterizing {Power} {Management} {Opportunities} for {LLMs} in the {Cloud}},
	isbn = {9798400703867},
	url = {},
	doi = {10.1145/3620666.3651329},
	language = {en},
	urldate = {2026-05-07},
	booktitle = {ASPLOS},
	publisher = {ACM},
	author = {Patel, Pratyush and Choukse, Esha and Zhang, Chaojie and Goiri, Inigo and Warrier, Brijesh and Mahalingam, Nithish and Bianchini, Ricardo},
	month = apr,
	year = {2024},
}

@misc{radovanovic_carbon-aware_2021,
	title = {Carbon-{Aware} {Computing} for {Datacenters}},
	url = {},
	doi = {10.48550/arXiv.2106.11750},
	urldate = {2026-04-30},
	publisher = {arXiv},
	author = {Radovanovic, Ana and Koningstein, Ross and Schneider, Ian and Chen, Bokan and Duarte, Alexandre and Roy, Binz and Xiao, Diyue and Haridasan, Maya and Hung, Patrick and Care, Nick and Talukdar, Saurav and Mullen, Eric and Smith, Kendal and Cottman, MariEllen and Cirne, Walfredo},
	month = jun,
	year = {2021},
	note = {arXiv:2106.11750},
}

@article{hall_carbon-aware_2025,
	title = {Carbon-{Aware} {Computing} for {Data} {Centers} with {Probabilistic} {Performance} {Guarantees}},
	issn = {1558-0679},
	url = {},
	doi = {10.1109/TPWRS.2025.3630499},
	urldate = {2026-04-30},
	journal = {IEEE Transactions on Power Systems},
	author = {Hall, Sophie and Micheli, Francesco and Belgioioso, Giuseppe and Radovanović, Ana and Dörfler, Florian},
	year = {2025},
}

@article{skrjanc_systematic_2023,
	title = {A systematic literature review on under-frequency load shedding protection using clustering methods},
	volume = {180},
	issn = {13640321},
	url = {},
	doi = {10.1016/j.rser.2023.113294},
	language = {en},
	urldate = {2026-04-22},
	journal = {Renewable and Sustainable Energy Reviews},
	author = {Skrjanc, T. and Mihalic, R. and Rudez, U.},
	year = {2023},
	pages = {113294},
}

@inproceedings{hespanhol_structure_2019,
	title = {A {Structure} {Exploiting} {Branch}-and-{Bound} {Algorithm} for {Mixed}-{Integer} {Model} {Predictive} {Control}},
	isbn = {9783907144008},
	url = {},
	doi = {10.23919/ECC.2019.8796242},
	urldate = {2026-04-30},
	booktitle = {ECC},
	publisher = {IEEE},
	author = {Hespanhol, Pedro and Quirynen, Rien and Di Cairano, Stefano},
	month = jun,
	year = {2019},
}

@incollection{lindbergGuideReducingCarbon2021,
  title = {A {{Guide}} to {{Reducing Carbon Emissions}} through {{Data Center Geographical Load Shifting}}},
  booktitle = {ACM e-Energy},
  author = {Lindberg, Julia and Abdennadher, Yasmine and Chen, Jiaqi and Lesieutre, Bernard C. and Roald, Line},
  year = 2021,
  month = jun,
  doi = {10.1145/3447555.3466582}
}

@article{lechowiczLearningAugmentedCompetitiveAlgorithms2025b,
  title = {Learning-{{Augmented Competitive Algorithms}} for {{Spatiotemporal Online Allocation}} with {{Deadline Constraints}}},
  author = {Lechowicz, Adam and Christianson, Nicolas and Sun, Bo and Bashir, Noman and Hajiesmaili, Mohammad and Wierman, Adam and Shenoy, Prashant},
  year = 2025,
  month = mar,
  journal = {SIGMETRICS},
  volume = {9},
  number = {1},
  publisher = {Association for Computing Machinery},
  doi = {10.1145/3711701}
}

@misc{lin_slasher_2026,
	title = {Slasher: {Power} {Flexibility} for {Cloud} {Datacenters}},
	shorttitle = {Slasher},
	doi = {10.48550/arXiv.2608.26021},
	urldate = {2026-08-27},
	publisher = {arXiv},
	author = {Lin, Liuzixuan and Kazhamiaka, Fiodar and Kumbhare, Alok Gautam and Zhang, Chaojie and Wang, Jaylen and Khan, Hassan and Assis, Rodrigo L. and Rodrigues, Mariana and Woolcock, Kyle and Mahalingam, Nithish and Warrier, Brijesh and Fonseca, Rodrigo and Bianchini, Ricardo},
	month = aug,
	year = {2026},
	note = {arXiv:2608.26021 [cs.DC]},
}

@article{palomar_tutorial_2006,
	title = {A tutorial on decomposition methods for network utility maximization},
	volume = {24},
	copyright = {https://ieeexplore.ieee.org/Xplorehelp/downloads/license-information/IEEE.html},
	issn = {0733-8716},
	doi = {10.1109/JSAC.2006.879350},
	number = {8},
	urldate = {2026-09-01},
	journal = {IEEE Journal on Selected Areas in Communications},
	author = {Palomar, D.P. and {Mung Chiang}},
	month = aug,
	year = {2006},
	pages = {1439--1451},
}

@article{takapoui_simple_2020,
	title = {A simple effective heuristic for embedded mixed-integer quadratic programming},
	volume = {93},
	issn = {0020-7179, 1366-5820},
	doi = {10.1080/00207179.2017.1316016},
	language = {en},
	number = {1},
	urldate = {2026-09-01},
	journal = {International Journal of Control},
	author = {Takapoui, Reza and Moehle, Nicholas and Boyd, Stephen and Bemporad, Alberto},
	month = jan,
	year = {2020},
	pages = {2--12},
}

@inproceedings{chen_dynamic_2013,
	title = {Dynamic server power capping for enabling data center participation in power markets},
	issn = {1558-2434},
	doi = {10.1109/ICCAD.2013.6691107},
	urldate = {2026-09-01},
	booktitle = {ICCAD},
	author = {Chen, Hao and Hankendi, Can and Caramanis, Michael C. and Coskun, Ayse K.},
	month = nov,
	year = {2013},
}

@inproceedings{badiei_diba_2016,
	title = {{DiBA}: {Distributed} {Power} {Budget} {Allocation} for {Large}-{Scale} {Computing} {Clusters}},
	shorttitle = {{DiBA}},
	doi = {10.1109/CCGrid.2016.101},
	urldate = {2026-09-01},
	booktitle = {CCGrid},
	author = {Badiei, Masoud and Zhan, Xin and Azimi, Reza and Reda, Sherief and Li, Na},
	month = may,
	year = {2016},
}

@article{boyd_distributed_2010,
	title = {Distributed {Optimization} and {Statistical} {Learning} via the {Alternating} {Direction} {Method} of {Multipliers}},
	volume = {3},
	issn = {1935-8237, 1935-8245},
	doi = {10.1561/2200000016},
	language = {en},
	number = {1},
	urldate = {2025-07-17},
	journal = {Foundations and Trends® in Machine Learning},
	author = {Boyd, Stephen and Parikh, Neal and Chu, Eric and Peleato, Borja and Eckstein, Jonathan},
	year = {2010},
	pages = {1--122},
}

@article{patriksson_survey_2008,
	title = {A survey on the continuous nonlinear resource allocation problem},
	volume = {185},
	issn = {0377-2217},
	doi = {https://doi.org/10.1016/j.ejor.2006.12.006},
	number = {1},
	journal = {European Journal of Operational Research},
	author = {Patriksson, Michael},
	year = {2008},
	pages = {1--46},
}

@book{european_commission_flexibility_service,
	title = {Flexibility as a service},
	doi = {doi/10.2926/6413481},
	publisher = {Publications Office of the European Union},
    author = {{European Commission} and {Directorate-General for Energy} and {European Climate, Infrastructure and Environment Executive Agency} and Volpe, R. and Bernstrauch, B. and Stuxberg, A. and Bianchi, F. and Chilou, T. and Ganesan, K. and Hatziargyriou, N.},
	year = {2026},
}

\begin{IEEEbiography}[{\includegraphics[width=1in,height=1.25in,clip,keepaspectratio]{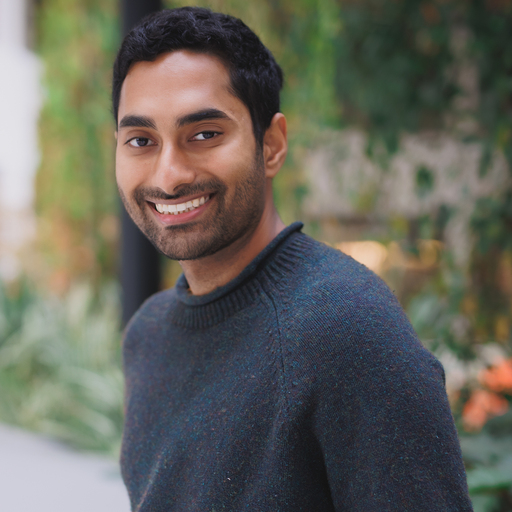}}]{Akshay Sreekumar} received the B.S. degree in electrical engineering and computer sciences from the University of California, Berkeley, Berkeley, CA, USA, and the M.S. degree in electrical and computer engineering from the University of California, Los Angeles, Los Angeles, CA, USA. He is currently a Ph.D. candidate in electrical engineering at Stanford University, Stanford, CA, USA, where he is a Burt and Deedee McMurtry Fellow. His research is at the intersection of large-scale optimization, learning, and control based methods for improving the efficiency, reliability, and operation of energy and computing infrastructure.
\end{IEEEbiography}

\begin{IEEEbiography}[{\includegraphics[width=1in,height=1.25in,clip,keepaspectratio]{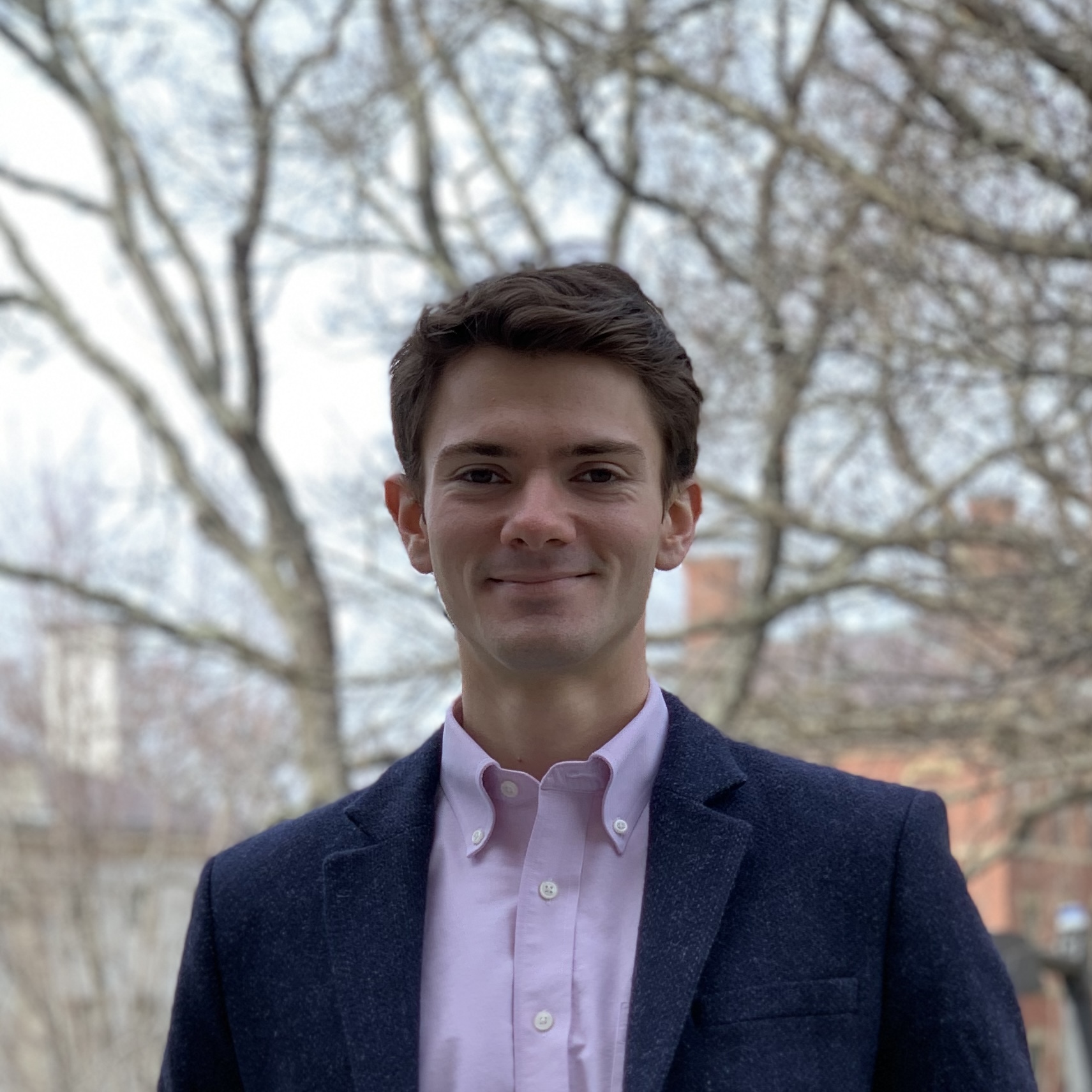}}]{Nicolas Christianson} received the A.B. degree in applied mathematics from Harvard University, Cambridge, MA, USA, in 2020, and the Ph.D. degree in computing and mathematical sciences from the California Institute of Technology, Pasadena, CA, USA, in 2025. He is currently an Assistant Professor in the Department of Computer Science at Johns Hopkins University; previously, he was a Stanford Energy Postdoctoral Fellow at Stanford University. His research lies at the intersection of algorithms, machine learning, and optimization, with applications to energy and computing systems. Dr. Christianson was the recipient of the California Institute of Technology's Ben P.C. Chou Doctoral Prize in Information Science and Technology and the ACM SIGEnergy Doctoral Dissertation Award.
\end{IEEEbiography}

\begin{IEEEbiography}[{\includegraphics[width=1in,height=1.25in,clip,keepaspectratio]{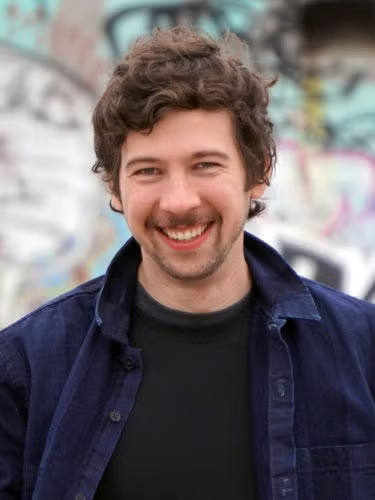}}]{Fiodar Kazhamiaka} is a researcher at Microsoft, affiliated with the Azure Research – Systems group, where he works on cloud efficiency and sustainability.
His research focus is to reduce datacenter carbon emissions and costs through server design and power/resource management, and developing interfaces between datacenters and the electricity grid.
He has published on topics spanning large-scale resource allocation and scheduling for cloud workloads, and carbon-aware design of datacenters and servers.
Fiodar Kazhamiaka received a PhD in Computer Science from the University of Waterloo.
\end{IEEEbiography}

\begin{IEEEbiography}[{\includegraphics[width=1in,height=1.25in,clip,keepaspectratio]{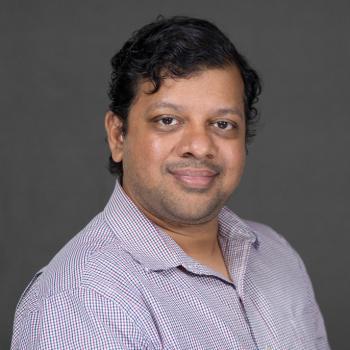}}]{Ram Rajagopal} is an Associate Professor of Civil and Environmental Engineering and Electrical Engineering at Stanford University, where he directs the Stanford Sustainable Systems Lab (S3L). Ram received his Ph.D in Electrical Engineering and Computer Sciences and M.A. in Statistics from the University of California, Berkeley. He is a recipient of the NSF CAREER Award, Powell Foundation Fellowship, Berkeley Regents Fellowship and the Makhoul Conjecture Challenge award.
\end{IEEEbiography}

\end{document}